\documentclass[10pt, conference, letterpaper]{elsarticle}
\usepackage{amsmath,amssymb,amsfonts}
\usepackage{graphicx}
\usepackage{textcomp}
\usepackage{xcolor}
\usepackage{amsthm}
\usepackage{paralist}
\usepackage{algorithm}
\usepackage{algpseudocode}
\usepackage{mathtools}
\usepackage{caption}
\usepackage{subcaption}
\usepackage{balance}

\DeclareMathOperator{\Paths}{\mathcal{P}}
\newcommand{\bw}{\text{bw}}
\newcommand{\capa}{b}
\DeclareMathOperator{\dem}{d}
\newcommand{\card}[1]{\ensuremath{\left|#1\right|}}
\DeclareMathOperator{\lb}{\textit{lb}}
\newcommand{\Demands}{\mathcal{D}}

\DeclarePairedDelimiter\floor{\lfloor}{\rfloor}

\newtheorem{rem}{Remark}

\newtheorem{theorem}{Theorem}

\newcommand{\rev}[1]{{\color{black}#1}}
\newcommand{\seb}[1]{{\color{black}#1}}
\newcommand{\jeremie}[1]{{\color{black}#1}}

\algnewcommand{\algorithmicforeach}{\textbf{for each}}
\algdef{SE}[FOR]{ForEach}{EndForEach}[1]
  {\algorithmicforeach\ #1\ \algorithmicdo}
  {\algorithmicend\ \algorithmicforeach}

\graphicspath{{graphics/}}

\def\BibTeX{{\rm B\kern-.05em{\sc i\kern-.025em b}\kern-.08em
    T\kern-.1667em\lower.7ex\hbox{E}\kern-.125emX}}
\begin{document}
\begin{frontmatter}
\title{Joint Routing and Scheduling for Large-Scale Deterministic IP Networks}

\author[label1]{Jonatan Krolikowski}
\author[label1]{S\'ebastien Martin}
\author[label1]{Paolo Medagliani}
\author[label1]{J\'er\'emie Leguay}
\author[label2]{Shuang Chen}
\author[label2]{Xiaodong Chang}
\author[label2]{Xuesong Geng}
\address[label1]{Paris Research Center, 20 Quai du Point du Jour, 92100 Boulogne-Billancourt, France}
\address[label2]{Beijing Research Center, Huawei Technologies Co. Ltd}

\begin{abstract}
With the advent of 5G and the evolution of Internet protocols, industrial applications are moving from vertical solutions to general purpose IP-based infrastructures that need to meet deterministic Quality of Service (QoS) requirements. The IETF \emph{DetNet} working group aims at providing an answer to this need with support for (i) deterministic worst-case latency and jitter, and (ii) zero packet loss for time-sensitive traffic. 

In this paper we focus on the joint routing and scheduling problem in large scale deterministic networks using Cycle Specified Queuing and Forwarding (CSQF), an extension of Cyclic  Queuing and Forwarding (CQF) with multiple transmission queues and support of segment routing.
In this context, we present two centralized algorithms to maximize traffic acceptance for network planning and online flow admission.
We propose an effective solution based on column generation and dynamic programming. Thanks to the reinforcement of the model with valid inequalities, we improve the upper bound and the solution. We demonstrate on realistic instances that we reach an optimality gap smaller than 10\% in a few seconds. Finally, we also derive an ultra-fast adaptive greedy algorithm to solve the problem at the cost of a small extra gap.
\end{abstract}


\end{frontmatter}

\section{Introduction}

The $5^{th}$ generation of networks is paving the road for 
latency-sensitive network services 
to enable a wide-range of applications like factory automation, connected vehicles and smart grids~\cite{rfc8578}.
Traditional Internet Protocol (IP) services allow delivering packets with no loss and no ordering issues. However, they cannot provide strict Quality of Service (QoS) guarantees. Certain service classes can be given preferential treatment but performance is still statistical. Deterministic performances are now a must to support applications with low and worst-case latency requirements\rev{, such as audio and video bridging, industrial automation (smart factory), smart grid, and remote control for telemedicine or automotive}.

A collection of IEEE 802.1 Ethernet standards, known as \emph{Time-Sensitive Networking (TSN)}~\cite{TSNSurvey19}, has been developed in the past decade to support professional applications over Local Area Networks (LAN) with mechanisms such as priority queuing, preemption, traffic shaping and  time-based opening of gates at output ports. While these mechanisms are well suited for static traffic requirements and small networks, they are not enough to support large-scale IP networks. The IETF DetNet (Deterministic Networking)~\cite{rfc8578} working group is taking a step
further by defining Segment Routing (SR) mechanisms so that Layer~3 can dynamically exploit Layer~2 functionalities for queuing and scheduling 
to support (i) deterministic worst-case latency and jitter, and (ii) zero packet loss for time-sensitive traffic. In particular, 
the working group is currently specifying Cycle Specified Queuing and Forwarding (CSQF)~\cite{chen-detnet-sr-based-bounded-latency-00},
a promising extension to Cyclic Queuing and Forwarding (CQF, a.k.a. IEEE 802.1Qch) with more than $2$ transmission queues in order to relax tight time-synchronization constraints and to schedule, in a more flexible way, transmissions at each hop.

In TSN layer-2 networks, several works have optimized the opening and closing of gates at output ports (IEEE 802.1Qbv, \cite{Pop2016,Craciunas2016,Mahfouzi2018}) to meet low latency requirements.
However, these solutions suffer from two main limitations: 
1) the overall gate schedule has to be modified at network-level every time the traffic characteristics evolve and 2) no queues can be used to dynamically delay packets at nodes. Alternatively, CSQF proposes a scalable solution where transmission cycles at each port repeat periodically thanks to the round-robin opening of multiple queues  dynamically selected by IP packets using 
segment routing identifiers (SIDs), 
a label stack that determines scheduling and routing at each hop. A network controller decides the proper label stack for each flow by solving a joint scheduling and routing problem. In this context, we propose two centralized control plane algorithms to maximize traffic acceptance both in the \emph{offline} (i.e., global optimization) and the \emph{online} (i.e., fast demand acceptance) scenarios. \rev{We point out that the offline algorithm can be used either to dimension the network, i.e., configuring network parameters, or to route batches of demands into the network, while the online algorithm can be used to quickly accept new demands as soon as they arrive into the network considering the same network parameters decided by the offline planning.} \jeremie{For the online algorithm, the inputs to the problem are discovered in sequence and a decision must be taken at demand arrival, while for the offline algorithm to dimension network resources all inputs are known when solving the problem.}
Up to our knowledge, it is the first paper to formulate the joint routing and scheduling problem for DetNet and provide efficient algorithms in large-scale deterministic networks.

We formulate the Deterministic Networking (DN) planning problem to maximize the acceptance of time-triggered traffic and we analyze its NP-hardness.
Then, we present an effective solution based on column generation and dynamic programming to solve a relaxed version of the problem which we round afterwards. Furthermore, thanks to the reinforcement of the problem model with valid inequalities, we show that we can drastically improve the upper bound (by up to 30\%) to better estimate the optimality gap and enhance the final solution (by up to 5\%) on large instances.
We demonstrate on realistic instances with hundreds of nodes and links that we can reach a gap smaller than 10\% in a few seconds. 
Finally, as an alternative, we derive an ultra-fast adaptive greedy algorithm (10 $\mu s$ per demand) at the cost of an extra 5\% gap when compared to the advanced solution based on column generation. This algorithm can be used for the quick acceptance of new demands in an online fashion.


More details about CSQF are given in Sec.~\ref{model}.
Relevant related works are discussed in Sec.~\ref{related}. The DN problem is formulated and analyzed in Sec.~\ref{problem}. Sec.~\ref{solution} derives the column-generation algorithm and Sec.~\ref{online} presents the adaptive greedy solution. Numerical assessments are shown in Sec.~\ref{results}. Sec.~\ref{conclusion} concludes this paper.

\section{DetNet System Model with CSQF}
\label{model}



A promising standard draft from the IETF DetNet group is the Cycle Specified Queuing and Forwarding (CSQF) mechanism~\cite{chen-detnet-sr-based-bounded-latency-00}. 
It is an evolution of the Cyclic Queuing and Forwarding (CQF)~\cite{7961303}, also referred to as peristaltic shaper, which considered 2 queues on ports, open and closed alternatively in a cyclic fashion. At any given time, one queue is for transmission while the other one is for reception. CQF works well for small networks as it assumes perfect synchronization between nodes and as the delay of packets cannot be dynamically controlled. A packet sent from a node in a cycle $c$ must be received during the same cycle and retransmitted at cycle $c+1$. To improve scalability and flexibility, DetNet CSQF adds the possibility of using more queues 
for loose synchronization between nodes and advanced scheduling~\cite {finn-detnet-bounded-latency-04,qiang-detnet-large-scale-detnet-04}. \rev{Differently from CQF which is a layer 2 protocol, CSQF operates at layer 3 as it allows the routing and the scheduling of packets using Segment Routing (SR).}
\subsection{Segment Routing for packet forwarding}
\rev{While in CQF a packet can only be forwarded at the next transmission cycle that follows the reception one, CSQF allows a flexible transmission scheduling by using a SR label stack to explicitly state for each intermediate node on which port (\emph{routing}) and in which cycle, i.e., which queue (\emph{scheduling})}, each packet should be transmitted after being received and processed. Precise knowledge of the position of a packet inside the network at a generic instant $t$ comes from the fact that, at each node, the worst case forwarding latency is known. Each time a packet arrives at a node, the scheduling of its future transmission is realized by its assignment to one of the inactive queues. \rev{As shown in Figure~\ref{csqfSR}, the SR Id (SID), i.e., the label, contained in the header of the packet allows at each SR-enabled node to first determine the output port for the packet. A range of SIDs is assigned to each port in order to define the outgoing queue for each packet. Once the port is chosen, the SID is used to map the corresponding outgoing cycle. Second, as the cycles are statically mapped into transmission queues, the SID is used to determine the outgoing queue in which the packet will be transmitted. The role of the centralized controller is to define, for each flow, on which port and in which queue each packet will be inserted, in order to avoid congestion on queues. As the number of outgoing cycles depends on the life cycle of each flow, a simplification consists on considering a finite number of cycles on which the traffic can be mapped. 

Like in TSN, DetNet traffic with CSQF is time-triggered (TT) and follows a specific pattern that repeats over time. This period is referred to as \emph{hypercycle}.  For each cycle, the application specifies how much data will be sent. To ensure deterministic end-to-end performance, it is necessary to provide a scheduling and routing decision at each hop and guarantee that enough capacity is available.}
\begin{figure}[!t]
	\centering
	\includegraphics[width=0.9\columnwidth]{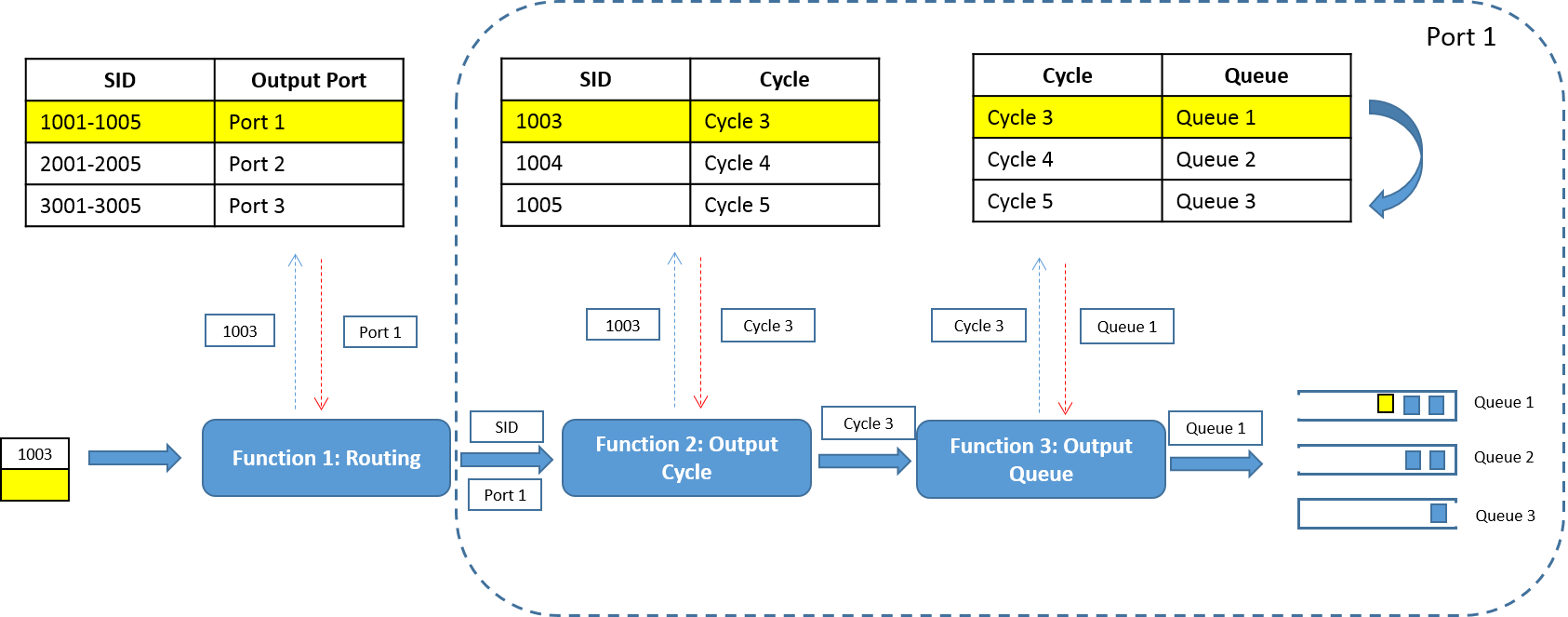}
	\caption{CSQF packet forwarding based on SR headers.
	}
	\label{csqfSR}
\end{figure}
According to CSQF, a DetNet-enabled device decides how and when a packet is forwarded by consuming the first Segment Routing ID (SID) available in the label stack of packet headers. As a first step, the receiving node maps the SID into the corresponding output port. As a second step, the device uses the same label to select the queue associated with the intended transmission cycle. \jeremie{MPLS or IPv6 can be used. In the first case, labels are encoded over 32 bits, in the former case they are basically IPv6 addresses (e.g., 128 bits when no compression is used).} The SR label stack can be provided by a centralized network controller that (a) computes a feasible path from source node to destination node, (b) computes the right scheduling within each node traversed by the flow, and (c) distributes the corresponding SR label stack to all the network elements via specific protocols (e.g., PCEP). 


\subsection{Deterministic forwarding in DetNet devices}
\rev{The delay introduced by a node to forward a packet can be split into 4 terms: (i) the propagation delay, (ii) the processing delay, (iii) the transmission delay, and (iv)   the queuing delay.  The propagation delay is given by the physical distance between two entities, the processing delay is the time required to receive the packet and sent it to the upper layers of the ISO/OSI stack for routing and scheduling decision, the transmission delay to put the packet on the physical link. While the propagation and the forwarding delay can be assumed as constant, the processing delay can vary due to different reasons. In order to provide deterministic forwarding latency, it is possible to measure the worst-case processing latency and use the queuing delay to compensate the processing delay, ensuring that the sum of the two terms is equal to a constant and known value---If the processing delay is large, the time spent inside a queue will be small and the packet will be scheduled for quick transmission, while if the processing delay is small, the time spent in the queue will be large and the packet will wait a for longer before being sent. Using this method, the processing delay, which is a stochastic process, is bounded by a constant value and a deterministic forwarding delay can be provided by the CSQF-enabled device. }

Inside a DetNet-enabled device, each port is equipped with $N$ queues (typically $8$), normally used for DiffServ and Best Effort (BE) traffic. In CSQF, the standard defines that out of the $N$ queues, $N_{\rm DN}$ queues (by default $3$) are reserved for time-sensitive traffic. These queues are served in a round-robin fashion such that the active queue is \emph{open} for transmission and \emph{closed} for reception. Conversely, the $N_{\rm DN}-1$ inactive queues can only accept packets for future transmission\rev{, i.e., a packet can be delayed by at most $N_{\rm DN}-1$ cycles, according to the queue in which the packet is inserted.} For this reason, the assignment of packets to specific inactive queues defines their transmission schedule and needs to be carefully controlled. 
Each time-sensitive queue is drained after the activity period and is dimensioned to receive all the packets scheduled within a cycle without introducing any packet loss.
And in order to support BE traffic, a percentage (e.g., 50\%) of the cycle duration is allocated to DetNet traffic while the remaining is for BE traffic. 
Due to the periodic activation of queues, the time at each node is logically divided into cycles. In order to guarantee deterministic latency, the duration of all cycles is the same throughout the network. The starting time of the cycles at the different nodes is not synchronized and can present an offset which is measured and known by the controller. 

\subsection{Deterministic packet forwaring: a networking view }

\rev{As in CSQF the forwarding delay is known, as well as the offset between nodes and the activation time of each queue, the controller can decide for the routing and scheduling of each flow in the network, ensuring that no collision or congestion can happen in the network. This is equivalent to deciding, for each packet, when and where it will be transmitted as well as its scheduling, i.e., if a packet is sent in the first available slot or delayed by one or more additional cycles before transmission. }


\begin{figure}[!t]
	\centering
	\includegraphics[width=0.9\columnwidth]{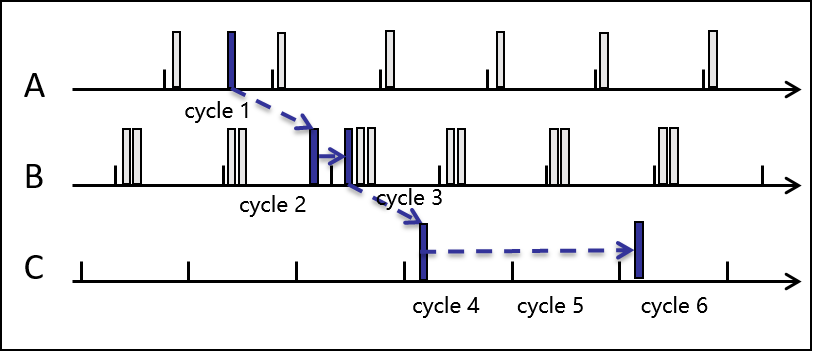}
	\caption{CSQF packet forwarding. Between nodes A and B, and, B and C, the packet is transmitted in the next cycle, while node C decides to schedule packet transmission two cycles later.
	}
	\label{csqfForwarding}
\end{figure}

In Figure~\ref{csqfForwarding}, we show an example of how a packet is propagated from node A to node C through node B. Once the packet is sent from A, it is received at B within a cycle (cycle 2 in the figure). As node B decides for immediate packet forwarding, the packet is transmitted in the next cycle. Finally, node C decides for the scheduling of the packet two cycles later, so that the packet will be transmitted at cycle 6. As the same considerations apply if we consider 0 offset between cycles of different nodes, for the sake of simplicity and without loss of generality, we will consider throughout this paper a 0 time offset such that all cycles are aligned at the different nodes.

\section{Related work}
\label{related}


In the literature, most of the papers are focusing on the scheduling of TSN/IEEE 802.1Qbv gate openings and closings to satisfy a certain traffic matrix. The matrix is composed of TT traffic flows which generate packets at known and repeating time instants. Routing information is generally given by the spanning tree protocol operating at layer 2. In this context, the goal is to find a feasible scheduling while minimizing the number of queues. In this case, a variant of the flow shop scheduling problem must be solved.  

For 802.1Qbv, \cite{Pop2016} introduces the problem as an Integer Linear Program (ILP) while~\cite{Craciunas2016} uses OMT (Optimization Modulo Theory) to formulate a Satisfiability Problem (SAT). \cite{Mahfouzi2018} also presents a SAT problem but considering robustness to control worst-case performance in case of uncertain traffic inputs. These papers do not introduce practical and efficient heuristics. The resolution of ILP or SAT models with solvers can only be achieved on very small instances.




In case routing can also be decided, \cite{Nayak2018} presents an online heuristic for 802.1Qbv. In this case, the end-to-end transmission of a cyclic TT flow must be realized in the same global transmission cycle to minimize the end-to-end latency. In other papers from the same authors, an ILP model is formulated to maximize traffic acceptance for a set of flows~\cite{Nayak2016}. \cite{Falk2018} formulates a similar problem by considering constant time shifts between incoming and outgoing transmissions at intermediate nodes (no controllable queuing is allowed). \cite{Smirnov2017} presents a SAT problem formulation of the same problem.
\rev{\cite{schweissguth2017ilp} proposes a compact ILP formulation of the joint routing and scheduling problem with the objective of minimizing the average latency. No scalable resolution algorithms are provided.}


Instead, our work focuses on both deterministic latency and jitter requirements rather than minimum latency. Our solution uses the recent CSQF standard proposal to guarantee worst-case performance at each hop thanks to the use of cyclic transmissions and segment routing for dynamic scheduling. We formulate the joint routing and scheduling problem for DetNet to maximize traffic acceptance as an ILP. We analyze the hardness of the problem and solve it at large scale and with quantifiable optimality.

\section{Problem formulation and Complexity}
\label{problem}

This section introduces our model for the routing and scheduling of traffic in DetNet with CSQF. We formulate its Integer Linear Problem (ILP) and analyze the complexity. \rev{The notation used in the following is summarized in Table~\ref{table:abbrev}.}

\begin{table}[t]
\caption{Overview of notation}\label{table:abbrev}

\begin{centering}
\begin{tabular}{ll}
\hline
Symbol                                          & Definition                                                                                       \\ \hline
 $C$                      & \#cycles per hypercycle                                                                          \\
$G=(V,A)$                & G topology of: V network nodes, A: directed links                                                \\
$\Delta_{a}$ / $\capa_a$ & delay / per-cycle capacity of link a                                                             \\
 $\Demands$               & set of demands                                                                                   \\
$s^d$, $t^d$         & source and destination of demand $d$                                                             \\
$\Paths$ / $\Paths^{d}$  & set of s-paths (all / belonging to demand $d$)                                                   \\
$d(p)$                   & demand that s-path $p$ belongs to                                                                \\
$\bw^d_c$ / $\bw^d$      & packets ($\mathit{du}$s) emitted   at $s^d$   (in cycle $c$ /  one hypercycle) by $d$                     \\
$\bw^{\dem(p)}_{a,p}(c)$ & capacity (in \textit{du}s) needed on $a$ for  $d(p)$ during    $c$ if $p$ is chosen \\
 $\Delta^d$               & delay constraint for $d$                                                                         \\
$r_k^p$                  & cycle shifts at $k$-th node of path $p$                                                          \\
$R$                     & max cycle shifts at all nodes                                                                    \\
$y_{p}$                  & decision variable: 1 if $p$ is chosen, 0 otherwise                                               \\ \hline
\end{tabular}
\end{centering}
\end{table}

\subsection{\rev{Cycles, Topology, and Demands}}
\rev{
Thanks to CSQF, time is partitioned into \emph{cycles} of equal duration, e.g.\ 10 $\mu$s. Blocks of consecutive cycles form hypercycles of size $C$, e.g.\ comprising 12 cycles each. $C$ is chosen such that the all network behaviour is the same in each hypercycle as will be argued below. Without loss of generality, we assume that the cycles start at the same time across the network and the hypercycle length $C$ is the same on every port / link.  

Let us consider a network $G=(V,A)$. The nodes $v\in V$ represent DetNet-enabled routers or switches. The nodes are connected with data links represented by the (directed) link set $A\subseteq V\times V$. Each arc $a=(u,v)\in A$ induces a delay of $\Delta_{a}$ cycles which comprises its propagation delay as well as the processing and queuing delay at node $v$. Furthermore, each arc $a$ has a per-cycle capacity $\capa_a$ (in data units \textit{du}, fixed size in Bytes).

A given set of demands $\Demands$, i.e., a set of TT flows, needs to be routed through the network. Demand $d$ is defined by  
\begin{compactitem}
\item a source node $s^d\in V$ and a destination $t^d\in V$,
\item a deterministic pattern of packet arrivals that repeats in every hypercycle. In cycle $c$, the source node $s^d\in V$ of demand $d$ emits packets for a  total of $\bw^d_c\in\mathbb{Z}_{+}$ (in data units \textit{du}). 
Note that due to the repetition.
$\bw^d_c = \bw^d_{c\%C}$ for any $c\in\mathbb{Z}_+$. 
\item a maximum acceptable end-to-end delay (in cycles) denoted by $\Delta^d$.
\end{compactitem}
}

\subsection{Scheduled Paths}
For a demand $d$ to be \emph{accepted}, the central controller  needs to assign a unique  feasible \emph{scheduled path} (\emph{s-path}).
An s-path $p$ is a path in $G$, i.e.\ a sequence of arcs $(a_1, \ldots, a_{\card{p}})$ where arcs $a_k=(u_k,v_k)$ are such that $u_1=s^d$, $v_{\card{p}}=t^d$ and $v_k=u_{k+1}$ for $k=1,\ldots, \card{p}-1$, together with an integer sequence $(r_1^p, \ldots, r_{\card{p}-1}^p)$ where $r_k^p\in\mathbb{Z}_{\geq 0}$  indicates the number of cycle shifts at corresponding nodes $v_k$. A shift is an explicit additional delay (expressed in multiple of cycles) that is introduced at nodes to schedule data transmissions into a specific CSQF queue. \rev{While the modeling for the routing and scheduling problem is based on cycle shifts, the path is finally encoded with transmission queues at each hop (i.e., for each outgoing port / link a SR label is derived from the chosen cycle shift).}
If $c$ is the earliest possible cycle in which a packet may be forwarded from $v_k$ (recall that processing and queuing delays are included in the arc delay of the preceding arc), a cycle shift of $r_k$ means that  transmission is carried out in cycle  $(c+r_k)\%C$.
The maximum number of shifts at a node is $R=N_{\rm DN}-2$ where $N_{\rm DN}$ is the number of CSQF queues reserved to deterministic traffic. 
The introduction of cycle shifts allows to accept more traffic, as we will see later. 

\begin{figure}
\includegraphics[width=\columnwidth]{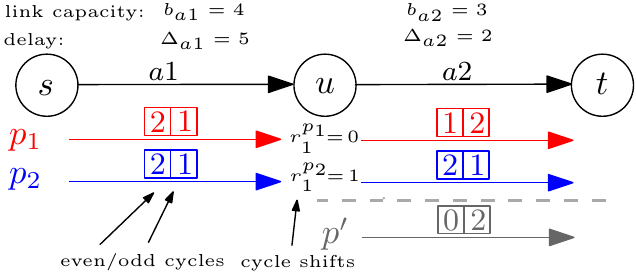}
\caption{A demand $d$ between source $s^d$ and destination $t^d$ with two \emph{s-paths} $p_1$ and $p_2$ and hypercycle length $C=2$, together with a second demand $d^{\prime}$ with only one path $p^{\prime}$.}
\label{fig:s-paths}
\end{figure} 
 
Let us consider the example in Figure~\ref{fig:s-paths} with a single demand $d$ where a path from $s$ to $t$  has two hops with an intermediate node $u$.  In all even cycles, $s$ sends $2$ \textit{du}, and in all odd cycles, it sends $1$ \textit{du}. The  pattern repeats after two cycles ($C=2$). We consider two possible s-paths denoted $p_1$ and $p_2$ that may be used for $d$. From $u$, the earliest  transmission of the 2 \textit{du}  is in the odd cycles ($c=5\%2=1$), while the one of the $1$ \textit{du} is in the even cycles (i.e., at $c=(1+5)\%2=0$), as the  delay is $\Delta_{a1}=5$ cycles. The s-path $p_1$ does not introduce any additional shift (recall that the one induced by the delay on $a_1$ is mandatory) at node $u$ ($r_1^{p_1}=0$).
 The s-path  $p_2$ has a cycle shift of $r_1^{p_2}=1$ at $u$, so it forwards $2$ \textit{du} in even cycles and $1$ \textit{du} in odd cycles. Now, to extend the example and show the need to introduce scheduling (i.e., extra cycle shifts) at intermediary nodes, let us consider a second demand $d^{\prime}$ from $u$ to $t$ that has only one available s-path $p^{\prime}$ due to delay constraints. This s-path uses arc $a_2$ such that in even cycles,  $d^{\prime}$ requires 0 \textit{du} on $a_2$ and 2 \textit{du} in odd cycles. S-paths $p_1$ and $p^{\prime}$ together thus require 4 \textit{du} in odd cycles, prohibitive with a per-cycle capacity $\capa_{a_2} = 3$. However, the additional cycle shift in $p_2$ for demand $d$ allows both $d$ and $d^{\prime}$ to be routed via $a_2$.

A single s-path is \emph{feasible} for demand $d$ if the following two conditions hold.\\
\textbf{1) End-to-end delay: }
The s-path delay $\Delta(p)$ must not exceed the maximum end-to-end delay $\Delta^d$. 
$\Delta(p)$ has two  aspects: (i) the sum of arc delays $\Delta_{a}$ and (ii) the sum of cycle shifts $r_k^p$ at the intermediate nodes. 
In the example in Figure~\ref{fig:s-paths}, $\Delta(p_1)= 7$, and $\Delta(p_2)= 8$, as indicated by the arc delays of 5 and 2, respectively. The difference comes from the shift at $u$ on $p_2$. 

We denote as $\Delta_{u_k}(p)$ the shift (in cycles) for the data to be transmitted at intermediate node $u_k$. It is  easily calculated as 
\begin{align}
\label{eq:shift}
\Delta_{u_k}(p) = \sum_{i=1}^{k-1} (\Delta_{a_{i}} + r^p_i)
\end{align}
where $\Delta_{u_1}(p)=0$ since there is no delay at $u_1=s^d$.
The total delay of the s-path is $\Delta(p)=\Delta_{u_{\card{p}}}(p) + \Delta_{a_{\card{p}}}$.\\
\textbf{2) Arc-cycle capacity:} 
A demand $d$ consumes a certain capacity $\bw^d_{a,p}(c)$ on  arcs $a$ of the s-path $p$ at cycle $c$. This value is determined by following the cyclic shifts along $p$: on the first arc of  $p$, the required capacity during cycle $c$ is $\bw^d_c$, the bandwidth emitted by the source $s^d$. 
As seen above, the delay at intermediate node $u$ is $\Delta_{u}(p)$. Any packet emitted from $s^d$ in cycle $c$ is thus forwarded from $u$ in cycle $(c+\Delta_{u}(p))\% C$.
The required bandwidth for demand $d$ during cycle $c$ on arc $a=(u,v)$ within s-path $p$ is therefore given by
\begin{align*}
\bw^d_{a, p}(c) = \bw^d_{(c+\Delta_{u}(p))\%C}.
\end{align*}
If enough capacity is available on every arc and during every cycle, $p$ can be assigned to $d$.


  

\rev{For ease of notation, each s-path $p$ is associated with a unique demand $\dem(p)\in\Demands$. Otherwise, two demands with identical sources and destinations would have identical s-paths variables. 
The set of feasible s-paths for a demand $d$ is denoted by $\Paths^{d}$.
$\Paths = \bigcup_{d} \Paths^{d}$ is the disjoint union of all s-paths variables.}

\begin{rem}
\label{rem:pathset}
The path set $\Paths$ is not given as an input. For each demand $d$, $\Paths^{d}$ needs to be generated. For general graphs, the cardinality of $\Paths^{d}$ may be exponential in the input size. 
\end{rem}

\subsection{Problem Statement}

The central controller tries to route each demand $d$ via a unique feasible s-path in $\Paths^{d}$. This is indicated with the variable $y_p$ which is set to $1$ if $p$ is chosen for $\dem(p)$, $0$ otherwise. Uniqueness of the s-path is ensured with the constraint 
\begin{align}
\label{eq:dn-unique}
\sum_{p\in\Paths^{d}} y_{p} \leq 1 \quad\forall d\in\Demands.
\end{align}

The arc capacities are shared among the routed demands. No more data than its capacity $\capa_a$ may be sent onto any arc $a$ during any cycle $c$.
This condition is ensured by the constraint
\begin{align}
\label{eq:dn-capa}
\sum_{p \in \Paths: a \in p} \bw^{\dem(p)}_{a,p}(c) \; y_{p} \leq \capa_{a} \quad \forall a\in A, \forall c.
\end{align}
Note that due to the given cyclic structure, it suffices to calculate the bandwidth  in the cycles  $0$ to $C-1$. Note also that cycle shifts $r_i^p>0$ may allow for otherwise incompatible demands to be transmitted via the same arc.

The aim of the central controller is to accept a subset of demands such that the total accepted bandwidth is maximized. The bandwidth $\bw^d$ of demand $d$ is the  sum of the bandwidth transmitted over the cycles $0,\ldots, C-1$, i.e.\ $\bw^d = \sum_{c=0}^{C-1} \bw_c^d$.
Thus, the Deterministic Networking (DN) problem can be formulated as an ILP in the following way: 
\begin{align*}
(\text{DN})& \max && \sum_{p\in\Paths} \bw^{\dem(p)} y_{p}\\
& \text{s.t.} &&  \sum_{p\in\Paths^{d}} y_{p} \leq 1 \quad\forall d,\\
&&& \sum_{p \in \Paths: a \in p} \bw^{\dem(p)}_{a,p}(c)\; y_{p} \leq \capa_{a} \quad \forall a, c,\\
&&& y_{p}\in\lbrace 0,1\rbrace \quad \forall p.
\end{align*}



\subsection{Complexity Analysis}
\label{sec:complexity}

The DN problem is an NP-hard optimization problem. This is due to Theorem 1 that shows NP-completeness for the decision counterpart, called DND.

\begin{theorem}
	\label{theorem:complexity}
	The \emph{DND} problem is NP-complete.
\end{theorem}

\begin{proof}
DND decides if, for a given threshold $\ell\in\mathbb{R}_{+}$, there is a feasible solution to DN  with objective value $\geq \ell$. The following reduction proof is based on the well-known $k$-Disjoint Paths ($k$DP)  problem \cite[Theorem 19.7]{KorteVygen2007}. 
We consider the (NP-complete) version of $k$DP  which decides if $k$ arc-disjoint paths can be found between nodes $s$ and $t$ in a directed graph $G$. This problem can be reduced to an instance of DND by setting the number of cycles to $C=1$ and $\bw^d_0 = 1$ for $k$ demands that all have source $s$ and destination $t$. The capacity of every arc $a$ is chosen to be $\capa_a = 1$.  Choosing $\ell=k$,  DND returns true if and only if there are $k$ arc-disjoint paths in $G$. Since all reduction steps are polynomial in the problem size, the NP-hardness proof is complete.
Furthermore, it is  clear that  DND belongs to NP since the validity of any solution can be checked in polynomial time. Thus, DND is NP-complete.
\end{proof}
 
In fact, there are two aspects which induce the "hardness" of DN: the number of cycles $C$ and the routing aspect, i.e.\ the multitude of available paths per demand.


\paragraph{Complexity due to routing}
The DN problem  generalizes  the unsplittable Multi-Commodity Flow problem (uMCF, also called Unsplittable Flow problem, see for example \cite{Guruswami2003}) \jeremie{that is at the core of all routing problems through the introduction of cycles and delays. DN is a temporal expansion of uMCF as  transmission cycles need to be decided on each link of scheduled paths. In addition, all commodities must experience a maximum end-to-end delay.} However, in general the  coefficients in objective function and constraints of uMCF are independent and not related as in DN (recall that the objective in DN is maximization of the bandwidth that is also used in the capacity constraint). Guruswami et al \cite{Guruswami2003} show that it is NP-hard to approximate uMCF  within $\card{E}^{1/2-\varepsilon}$ for any $\varepsilon>0$. Their proof, however, can easily be extended  to DN (with its related coefficients) with the same result even in the case of $C=1$.
\paragraph{Complexity due to cycles} If $C$ is part of the input, and not a priori bounded, DN  cannot even be efficiently  approximated in polynomial time (unless P=NP), i.e.\ there is no polynomial-time approximation scheme (PTAS). This is true even if the graph $G$ consists only of one single arc. Then, DN is equivalent  to the 0-1 Multidimensional Knapsack (01MK) problem (see \cite{Kaparis2008}): a 01MK instance  is transformed to a DN instance by multiplying each constraint such that the right-hand side (\textit{rhs}) is the least common multiple of the given \textit{rhs} values. If the number of constraints (given by $C$) is unbounded, there is no PTAS for 01MK \cite{Korte1981}.

Note that DN becomes weakly NP-hard (and thus solvable in pseudo-polynomial time) if the number of cycles $C$ and the set of feasible s-paths $\card{\Paths}$  are bounded (and can be computed in polynomial time) since the same is true for 01MK with bounded dimensions \cite{Korte1981}.

\seb{To conclude this section, we show that the DN is harder than the classical unsplittable Multi-Commodity Flow problem. Indeed, the capacity constraints correspond to a knapsack problem (weakly NP-hard) on each link, whereas, the capacity constraints for DN problem correspond to the  0-1 Multidimensional Knapsack problem (strongly NP-hard) on each link. }

\section{Scalable Global Algorithm}
\label{solution}

\seb{This section presents a solution to DN based on \emph{Column Generation} (CG) and Randomized Rounding (RR), a classic approach for intractable ILPs.
Because the DN model has an exponential number of variables it is not possible to solve it with a linear solver. However, we can use a CG procedure to generate a polynomial sub set of variables ensuring the optimality of the linear relaxation of the DN model (referred to as LDN). We then round the LDN solution to an integer solution using a randomized rounding algorithm which provides a high-quality and feasible solution to the original DN problem.
The optimal solution to the relaxed problem provides a \emph{upper bound} (UB) to DN and it can be used to evaluate the integrality gap. By strengthening the capacity constraints (see Sec.~\ref{inequalities}), we present an enhanced LDN formulation that helps to improve the CG-RR solution as well as the UB.}

\subsection{Solving the Linear Relaxation}
\label{sec:LR}
LDN relaxes the integrality constraints on the variables $y_p$.
It is well-known that linear programs (LPs) such as LDN can  be solved in polynomial time in terms of input size  \cite{Khachiyan1979}. However, as to Remark~\ref{rem:pathset}, the number of variables in DN in general is not polynomial in the input size which poses a problem solving LDN in practice. We overcome this problem by applying column generation~\cite{desaulniers2006column} to LDN.

\subsubsection{Column Generation}
We start with a \emph{restricted LP} which contains
only a subset of the variables of the so called \emph{master LP} LDN. \seb{This subset of variables is given by the greedy algorithm described in the next section.} By solving the \emph{pricing problem}, we decide whether there are variables that are currently not contained in the restricted LP but might improve the objective value. If no such variables can be found, the current subset of variables is guaranteed to be sufficient to solve the master LP optimally. Otherwise, the newly generated variables are added to the restricted LP and the process iterates. This method is based on LP duality (see for example \cite{Schrijver2003}). 

In the following, we consider a subset of s-paths $\Paths^{\prime}\subseteq\Paths$ \seb{that respect the end-to-end delay}. For ease of notation, we assume that for all $d\in\Demands$, there is an s-path $p\in\Paths^{\prime}$ such that $\dem(p)=d$. The induced restricted relaxation of DN is:
\begin{align}
(\text{LDN}^{\prime})& \max && \sum_{p\in\Paths^{\prime}} \bw^{\dem(p)} y_{p}\nonumber\\
& \text{s.t.} &&  \sum_{p\in\Paths^{\prime}:\dem(p)=d} y_{p} \leq 1 \quad\forall d,\label{eq:ldn-1}\\
&&& \sum_{p\in\Paths^{\prime}: a\in p} \bw^{\dem(p)}_{a,p}(c)\; y_{p} \leq \capa_{a} \quad \forall a, c,\label{eq:ldn-2}\\
&&& y_{p}\geq 0 \quad \forall p\in\Paths^{\prime}.\nonumber
\end{align}
Note that a feasible solution $(y_p^{\prime})$ to $\text{LDN}^{\prime}$ induces a feasible solution $(y_p)$ to LDN by setting $y_p = y_p^{\prime}$ for $p\in\Paths^{\prime}$ and $y_p = 0$ otherwise. 
If $(y_p^{\prime})$ is optimal for $\text{LDN}^{\prime}$, we can determine if the induced solution $(y_p)$ is optimal to LDN by considering the dual of $\text{LDN}^{\prime}$: 
\begin{align*}
(\text{D-LDN}^{\prime})&&& \min \sum_{d} \lambda_d + \sum_{a}\sum_{c}\capa_{a} \mu_{a,c}\\
\text{s.t.}&&&\hspace{-2pt} \lambda_{\dem(p)} + \sum_{a\in p}\sum_{c} \bw^{\dem(p)}_{a,p}(c) \mu_{a,c} \geq \bw^{\dem(p)} \;\forall p\in\Paths^{\prime},\\
&&& \lambda_d \geq 0 \quad \forall d,\\
&&& \mu_{a,c} \geq 0 \quad \forall a,c,
\end{align*}
where the dual variables $\lambda_d$ relate to primal constraints (Eq.~\eqref{eq:ldn-1}) and dual variables $\mu_{a,c}$
relate to constraints (Eq.~\eqref{eq:ldn-2}).

Let $\left(\left(\lambda_d^{\prime\ast}\right),\left( \mu_{a,c}^{\prime\ast}\right)\right)$ be an optimal solution  for $\text{D-LDN}^{\prime}$. If there exists a \emph{separating s-path} $p\in\Paths\setminus\Paths^{\prime}$ such that  
\begin{align}
\label{eq:sep-path}
\lambda_{\dem(p)}^{\prime\ast} + \sum_{a\in p}\sum_{c} \bw^{\dem(p)}_{a,p}(c) \mu_{a,c}^{\prime\ast} < \bw^{\dem(p)},
\end{align}
then the solution is infeasible to D-LDN, the dual of LDN. The problem $\text{D-LDN}^{\prime\prime}$ with $\Paths^{\prime\prime}=\Paths^{\prime}\cup\lbrace p\rbrace$ constitutes an improved approximation to D-LDN. If no such separating s-path exists, the solution is feasible to D-LDN and also optimal for DLN. 



Note that for LDN, the  latency constraint must be integrated in the pricing problem. To solve the pricing problem, an s-path fulfilling Eq.~\eqref{eq:sep-path} needs to be found if and only if one exists. 

\subsubsection{Generation of Separating s-Paths}

Given an optimal solution $\left(\left(\lambda_d^{\prime\ast}\right),\left( \mu_{a,c}^{\prime\ast}\right)\right)$  to $\text{D-LDN}^{\prime}$, an algorithm generating separating s-paths \seb{ respecting the end-to-end delay} has to determine for each demand $d \in \Demands$ if such separating s-paths  $\Paths^{\prime\prime}\subseteq\Paths\setminus\Paths^{\prime}$ exist. 
If yes, it should return (a subset of)  $\Paths^{\prime\prime}$, $\emptyset$ otherwise.

For each demand $d \in \Demands$,  finding the (delay constrained) shortest s-path $p$ in terms of path weight $\sum_{a\in p}\sum_{c} \bw^{\dem(p)}_{a,p}(c) \mu_{a,c}^{\prime\ast}$
solves the pricing problem. If solved optimally, it guarantees that a path is found if it exists. If the weight of the shortest path is strictly smaller than $\bw^{\dem(p)} - \lambda_{\dem(p)}^{\prime\ast}$, then we add the column (variable) associated with this path to the problem. If for all demands, no columns can be added, the CG procedure terminates.

\begin{figure}
\centering
\includegraphics[width=0.8\columnwidth]{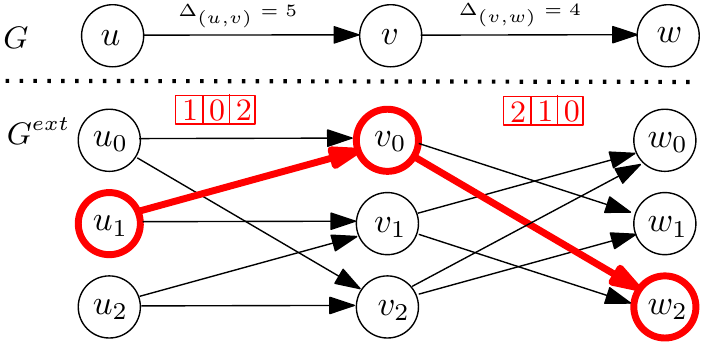}
\caption{Constructing the extended graph  $G^{\text{ext}}$ from graph $G$.}
\label{fig:extended}
\end{figure}

In order to compute a shortest s-path, we construct the extended graph $G^{\text{ext}}=(V^{\text{ext}},A^{\text{ext}})$ where $V^{\text{ext}} = \lbrace u_c \mid (u,c) \in V\times\lbrace 0, \ldots, C-1\rbrace\rbrace$.
When a path $p$ in $G^{\text{ext}}$ contains node $u_c$, the respective s-path  in $G$ passes the following arc $(u,v)$ with a cycle shift of $c$ w.r.t.\ the source $s^{d(p)}$ of the respective demand. The arc set $A^{\text{ext}}$ represents the possible transitions to the following node $v$. E.g.\ if s-path $p$ in $G$ has no additional shift scheduled at $v$, the path in $G^{\text{ext}}$ takes the arc $(u_c, v_{(c+\Delta_{(u,v)})\% C})\in A^{\text{ext}}$. Generally, if there is a scheduled shift of $r$ cycles at $v$, the respective arc is  $(u_c, v_{(c+\Delta_{(u,v)} + r)\% C})\in A^{\text{ext}}$ since the cycle in which $p$ leaves  $v$ is $(c+\Delta_{(u,v)} + r)\% C$.

As an illustration, Fig~\ref{fig:extended} shows three internal nodes $u,v,w$ of some path $p$ for a demand that has demand $2,1,0$ data units over cycles 0, 1 and 2, respectively ($C=3$). A maximum of 1 additional shift per node is allowed. Exiting node $u$, $p$ has a cycle shift of $1$, thus it contains $u_1$ in  $G^{\text{ext}}$. There is no additional shift at $v$, thus the following node in $G^{\text{ext}}$ is $v_{(1+\Delta_{(u,v)})\%3} = v_0$. At the following node there is a shift of one cycle ($r_{w}=1$), thus  $p$ contains $w_{(0+\Delta_{(v,w)}+1)\%3} = w_2$.

This construction allows  setting the arc weights in $G^{\text{ext}}$ independently of the specific path as $w_{(u_{c^{\prime}}, v_{c^{\prime\prime}})} = \sum_c \bw^d_{(c^{\prime} + c)\% C} \mu^{\prime\ast}_{a,c}$. Thus, finding a separating s-path in $G$ is equivalent to finding a simple path in $G^{\text{ext}}$ that respects both the weight and the delay constraint.

In case the end-to-end
 delay constraint is negligible, shortest path algorithms such as Dijkstra's may be applied to find the shortest path (in terms of arc weights) in polynomial time. 
In contrast, finding a  shortest path that also meets the delay constraint is NP-hard. In \cite{larac}, the authors propose a suboptimal but polynomial-time algorithm called LARAC based on the problem's Lagrangian relaxation and Dijkstra's algorithm. This heuristic may however not find any separating s-path even if one exists, making its use prohibitive for solving the pricing problem.

Our algorithm guarantees to find a separating s-path if one exists. Thus, LDN is solved to optimality and we obtain an upper bound to DN.
To efficiently solve this pricing problem, we apply a dynamic programming algorithm (see Algorithm~\ref{alg:s-paths}) that finds a suboptimal separating s-path for every demand $d$ in case it exists and guarantees to return an empty set in case no separating s-path exists. 
The algorithm reduces to a recursive depth-first search (DFS) on the extended graph $G^{\text{ext}}$ which can be in practice generated on the fly. 

For every node $v\in V$, we maintain a label $(w, \Delta)$ that signifies that $v$ has been reached by an s-path with an accumulated weight at most $w$ and latency at most $\Delta$. At any point during the execution of the algorithm we have $L(v_{c}) = L(v_{c^{\prime}})$ for all $c, c^{\prime}$. For every demand $d$, the label sets are initialized by $L(s^d)=\lbrace (0,0)\rbrace$ and $L(v)=\emptyset$ for $v\in V\setminus\lbrace v\rbrace$. 
 Let the current path at the node $u_c$ have a label $(w, \Delta)$. From $u_c$, the algorithm chooses a neighbor $v_{c^{\prime}}\in\delta^{\text{ext}}_{+}(u_c)$. The current path label is updated to $(w^{\prime}, \Delta^{\prime})=(w+w_{(u_c, v_{c^{\prime}})}, \Delta + \Delta_{(u_c, v_{c^{\prime}})})$ where $\Delta_{(u_c, v_{c^{\prime}})}$ is the delay of arc $a^{\prime}$ including the cycle shift. 
The new path is  rejected if the delay is too high, i.e.\ $\Delta + \Delta_{(u_c, v_{c^{\prime}})} > \Delta^d$, or if it is \emph{dominated}, i.e.\ if $L(v^{\prime})$ contains a label $(\bar{w}, \bar{\Delta})$ for which $\bar{w} \leq w^{\prime}$ and  $\bar{\Delta} \leq \Delta^{\prime}$. In this case, the algorithm goes back to $u_c$. Otherwise, the current label is added to $L(v^{\prime})$, and all labels in $L(v^{\prime})$ that are dominated by the current label are deleted.
If  destination node $t^d_c$ (for any $c$) is reached by the algorithm while the current delay $\Delta$ does not surpass the delay limit $\Delta^d $ and the current weight 
$w$ is smaller than $\bw^d-\lambda^{\prime\ast}_d$, add the corresponding path to the return set. 
Finally, return the set of all generated paths.

\begin{algorithm}[t]
\caption{$\texttt{generate-s-Paths}$}
\label{alg:s-paths}
\begin{algorithmic}
		\small
\State $\Paths^{\prime\prime} \coloneqq \emptyset$
\ForEach{$d\in\Demands$} 
	\State $w_{a,c} \coloneqq \sum_{c^\prime=0}^{C-1}\bw^d_{(c + c^{\prime})\% C} \mu^{\prime\ast}_{a,c^{\prime}} \ \ \  \forall a \in A, \forall c\in C $	
	\State $L(s^d)\coloneqq\lbrace(0,0)\rbrace$
	\State $u\coloneqq s^d, c\coloneqq 0, w\coloneqq 0$, $\Delta \coloneqq 0$
	\State $p \coloneqq \texttt{rec-s-Path}(u, c, w, \Delta, d)$
	\State $\Paths^{\prime\prime} \coloneqq \Paths^{\prime\prime}\cup\lbrace p \rbrace$
\EndForEach
\end{algorithmic}
\end{algorithm}

\begin{algorithm}
\caption{$\texttt{rec-s-Path}( u,c,w, \Delta, d)$}
\label{alg:rec-s-paths}
\begin{algorithmic}
	\small
\ForEach{$v\in\delta^{+}(u)$} \algorithmiccomment{iterate over outgoing arcs}
	\State $w \coloneqq w+ w_{(u,v),c}$ \algorithmiccomment{update weight}
	\State $\Delta \coloneqq \Delta + \Delta_{(u,v)}$ \algorithmiccomment{update delay}
	\If{$v=t^d$} \algorithmiccomment{destination reached}
		\If{$(w, \Delta)$ feasible}
			\State \Return $((u,v), \emptyset)$ \algorithmiccomment{accept arc, done}
		\Else
			\State reject $(u,v)$
		\EndIf
	\Else
		\ForEach{$r\in\lbrace 0, \ldots, R\rbrace$}	\algorithmiccomment{iterate over cycle shifts}
			\State $\Delta \coloneqq \Delta + r$ \algorithmiccomment{include shift}
			\State $c \coloneqq \Delta \% C$ \algorithmiccomment{update cycle}
			\If{$(w, \Delta)$ is feasible and not dominated} 
				\State delete all labels dominated by $(w, \Delta)$
				\State \Return $((u,v),r) + \texttt{rec-s-Path}(v, c, w, \Delta, d)$\\ \algorithmiccomment{accept arc and continue}
			\Else
				\State  reject current path
			\EndIf
		\EndForEach	
	\EndIf
\EndForEach
\end{algorithmic}
\end{algorithm}
Note that in the worst case, Algorithm~\ref{alg:s-paths} terminates after all $s^d$-$t^d$-paths  have been explored. However, in case of a tight delay bound, the algorithm is very fast. If it can be determined that the delay bound for demand $d$ is very permissive, the Algorithm~\ref{alg:s-paths} may be modified by reducing the bound $\Delta^d$ in a first run and, in case no path is found, iteratively increase it until its original value is reached. This procedure may avoid the enumeration of exponentially many paths.

\subsection{Randomized Rounding}

Once the optimal solution $(y_p^{\ast})$  to the linear relaxation LDN has been obtained, a feasible solution $y_p$ to DN is computed by \emph{randomized rounding}. \seb{As the linear relaxation provides a fractional solution where a demand can use several paths, in order to respect the uniqueness constraint for each demand, at most one path, out of those given by linear relaxation, must be selected. } 
For a demand $d$ picked at random, we assign a probability of $y_p^{\ast}/\sum_{p^{\prime}\in\Paths^d}y_{p^{\prime}}^{\ast}$ to each s-path $p\in\Paths^d$.
According to these probabilities, we choose a path $p\in\Paths^d$. If there is sufficient residual capacity in the network, we assign the s-path to demand $d$. Otherwise, delete the path, renormalize the remaining probabilities and iterate until an s-path is assigned or no s-path with positive probability remains. Then, we continue with the next demand. This algorithm is executed several times. The best solution, referred to as \emph{CG-RR} solution, is selected.

\subsection{Improving the Fractional Solution}
\label{inequalities}

In order to improve both the upper bound given by the linear relaxation and the \emph{CG-RR} solution, we leverage on the fact that in practice for any demand $d$, the required bandwidth $\bw^d_c$ per cycle $c$  is a multiple of a packet size $\textit{ps}^d$. While the packet sizes may vary among the demands, they are not arbitrarily distributed.
If the arc capacities $\capa_{a}$ are not multiples of the packet sizes, we can produce a fractional solution to DN that is closer to its optimal integer solution and thus improve the CG-RR solution as well as the upper bound by tightening the capacity constraints shown in Eq.~\eqref{eq:dn-capa}.




We denote the greatest common divisor of the bandwidth requirements $\bw^d_{a,p}(c)$ of all paths $p\in\bar{\Paths}^d$, $d\in\bar{\Demands}$ by $\textit{ps}_a$, where $\bar{\Demands}$ is the set of demands with at least one path through link $a$ with cycle shift $c$, $\bar{\Paths}^d$ is the respective set of paths for $d$ and $\bar{\Paths} = \bigcup_d \bar{\Paths}^d$.
We assume that $\textit{ps}_a$ is not a divisor of the arc capacity $\capa_{a}$.
Then 
capacity constraint is strengthened by division by $\textit{ps}_a$ for all $c$:
 \begin{align}
\label{eq:dn-capa2}
\sum_{p\in\bar{\Paths}: a\in p} \frac{\bw^{\dem(p)}_{a,p}(c)}{\textit{ps}_a} \; y_{p} \leq \Bigl\lfloor\frac{\capa_{a}}{\textit{ps}_a}\Bigl\rfloor
\end{align}
This constraint is valid  since the left hand side is integer, and it is stronger than that in Eq.~\eqref{eq:dn-capa} since $\floor{\frac{\capa_{a}}{\textit{ps}_a}} < \frac{\capa_{a}}{\textit{ps}_a}$.

\section{Fast Greedy Algorithm}
\label{online}

Alternatively to the CG-RR solution presented in Sec.~\ref{solution}, a more conventional \emph{Greedy} approach is to route the demands one-by-one. 
When a demand $d$ is next in line, the greedy algorithm tries to find a feasible s-path such that, for all affected arcs and all cycles, the capacity constraint is respected. 
We call such an s-path $\bar{\Paths}$-feasible where $\bar{\Paths}$ is the set of already assigned paths. 
If no such s-path can be found, $d$ is rejected, otherwise it is added to $\bar{\Paths}$. 
Such approach encompasses two subproblems: a) paths generation and b) path selection. The order of incoming demands is considered as input, such that the \emph{Greedy} algorithm can also be used in an online setting.

\subsection{Path Generation}
\label{sec:path-gen}
In order to generate $\bar{\Paths}$-feasible s-paths for any demand $d$, one can search for a set of $K$ maximally arc-disjoint paths and hope for a good load balancing \rev{(using the path selection algorithm in Sec.~\ref{pathselect})}. As for the IPRAN scenario described in Sec.~\ref{results} specific knowledge about the network allows to define sets of bottleneck arcs that should be mutually avoided, the algorithm can become more effective. In this context, we first identify sets of \emph{mutually avoidable} arcs $\lbrace a_1, \ldots, a_k\rbrace$ for which an arc can only belong to one s-path for $d$. In our scenario, the set is composed by outgoing and incoming arcs respectively at the source and destination nodes. 
Then, we use a shortest path algorithm with the delay as arc length (e.g. Dijkstra's algorithm) and a maximum delay and we enforce the use of exactly one of these arcs to generate diversified s-paths $\Paths^d$. \rev{In more general scenarios, we may use more general but slower K-maximally edge disjoint algorithms (see chapter 7.6 in~\cite{kleinberg2006algorithm}).}

The runtime of the algorithm depends in large parts on the path generation. Assuming a limitation $K$ on the number of generated paths per demand and an efficient Dijkstra implementation, the runtime is in $O(\card{\Demands}K(\card{A}+\card{V}\log(\card{V}))$.

\subsection{Path Selection}
\label{pathselect}
Given a set paths $\tilde{\Paths}^d$ of $\bar{\Paths}$-feasible s-paths for demand $d$, the simplest approach to path selection is assigning the first (or a random) $p\in\tilde{\Paths}^d$. 
However, this can lead to very low traffic acceptance as bottleneck links can quickly appear and partition the network.
To adress this problem, we use a form of load balancing inspired by competitive online routing algorithms~\cite{Awerbuch1993}.
For two feasible sets of s-paths $\Paths$, $\Paths^{\prime}$ for the same subset of demands, we consider a load balancing metric $\lb$ such that $\lb(\Paths) > \lb(\Paths^{\prime})$ if solution $\Paths$ is more balanced. Given a set of routed demands $\bar{\Paths}$ and demand $d$, \emph{Greedy} selects path  $p\in\tilde{\Paths}^d$ for which $\lb(\bar{\Paths} \cup \lbrace p\rbrace)$ is maximal. 
Based the idea of proportional fairness (see \cite{Kelly1998}), we use a load balancing metric $\lb$ that is maximal when the available bandwidth  on the arcs $A$ is fairly distributed:
\begin{align*}
\lb(\bar{\Paths}) = \sum_{a\in A} \log(\text{av}_a(\bar{\Paths}) + \varepsilon).
\end{align*}
where $\text{av}_a(\bar{\Paths})$ is the percentage of unused bandwidth on arc $a$ when paths $\bar{\Paths}$ are used.
The addition of a small $\varepsilon > 0$ allows for the case in which exhausting the capacity of some arc cannot be avoided. The percentage of unused bandwidth is defined for the busiest cycle, i.e.\
\begin{align*}
\text{av}_a(\bar{\Paths}) = 1 - \max_c \frac{\sum_{ p\in\bar{\Paths}:  a\in p} \bw^{\dem(p)}_{a,p}(c) y_p}{\capa_a}.
\end{align*}
This definition reflects that bandwidth should be kept available for future demands on all cycles in a fair manner.

%

\section{Numerical Evaluation}
\label{results}

This section presents results in a realistic 5G scenario computed with a C++ environment on a 40$\times$3.0 GHz machine with 190GB RAM. Linear programs are solved with IBM CPLEX 12.6.3.

\subsection{Setup}
\label{sec:eval-setup}

We consider a typical IPRAN (IP Radio Access Network) scenario with $1700$ nodes  connected via $5200$ directed arcs. The topology is divided into 3 layers: access, aggregation, and core. Access layer is composed of 1600 nodes, i.e. 800 BS (Base Station) and 800 CSG (Cell Site Gateway). The aggregation layer is composed of 80 nodes referred to as ASG (Aggregation Site Gateway). 
In the core layer there are 20 RSG (Radio Service Gateway) nodes connected to the EPC (Evolved Packet Core).
\rev{We choose this network topology for all evaluation since it is highly relevant for DetNet applications.}


The capacity of links in the access and aggregation are 10 Gbps and 40 Gbps, respectively. 
In the core, links have a capacity at either 100 Gbps or 400 Gbps.
Each BS has a 1-to-1 mapping with a CSG. Each CSG is connected to a pair of ASG via a direct link. Up to 20 CSG are connected to the same pair of ASG. Groups of ASG are connected via a ring with some additional shortcuts. A group of connected ASG and their CSG form a domain. There are 10 domains in the network. The core network is fully meshed.

The link delay is chosen proportionally to the distance between its nodes: for the access link it is uniformly distributed between 0.2 and 0.8 ms, corresponding to a distance of 10-40 km between elements. In the aggregation, the link delay is uniformly distributed between 0.8 and 1.6 ms, while in the core it is uniformly distributed between 2 and 10 ms. The cycle duration is 10 $\mu s$ and the internal processing delay (worst-case) is 30 $\mu s$ for each node.

We consider 250 to 2500 demands for each scenario. Each demand has a hypercycle $C=12$ and a packet size of 500  Bytes. We consider that the traffic pattern is binary: either there is some traffic sent in a cycle or there is no traffic at all. In case there is some traffic, we consider that either 1 or 2 packets are sent per cycle, that corresponds to a max throughput of 200 Mbps. The same number of packets is sent in every cycle with data transmission. We consider three traffic patterns  randomly selected: one data transmission every 2, every 3, or every 6 cycles. Demands are shifted at the beginning by a random number of cycles. 
\rev{In our main demand scenario (Sc1), 60\% of demands are directed to a BS which is connected to the same pair of ASG, via the associated CSG nodes, of the source node (labeled as $D_1$), 30\% of demands are directed to a BS which is in the same domain of the source node (labeled as $D_2$), and 10\% of demands are directed to a BS in a different domain (labeled as $D_3$). In secondary scenarios, the distributions over $D_1$, $D_2$ and $D_3$ are $100\% / 0\% / 0\%$ (Sc2) and $34\% / 33\% / 33\%$ (Sc3).} The end-to-end delay constraint is using a discrete uniform distribution between 1, 2, and 3 ms for $D_1$ demands, between 4, 5, and 6 ms for $D_2$ demands, and between 40, 50, and 60 ms for $D_3$ demands. 

To eliminate statistical fluctuations, results are obtained by averaging on 10 different traffic realizations.
In the considered scenario, \emph{Greedy} computes for each demand at most $K=4$ disjoint paths when $N_{\rm DN}=2$
and $K=8$ paths when $N_{\rm DN}=3$ to account for each possible time shift at the first CSG.

\subsection{\rev{Benchmark solutions}}
\label{sec:eval-setup}

\rev{We consider that each node is running the CSQF standard with $N_{\rm DN}=3$ queues that can be used for DetNet traffic. As a first benchmark solution, we consider the case in which $N_{\rm DN}=2$ for all nodes, corresponding to a CQF solution in which additional shifts at intermediary nodes are not possible. The two solutions are labeled to as \emph{CSQF} and \emph{CQF} in the plots.}

\rev{In addition, as a comparison point to our algorithms, we considered the case where the exact knowledge about the cycle-specific demand patterns is not available. In this case, a demand is characterized as a flow with a total arrival volume of data over each hypercycle. We still request the solution to uphold the DetNet guarantees on delay and jitter.
Since, in the worst case, all data packets are emitted by the source during the same cycle, the entire capacity needs to be reserved during all cycles on all links of the chosen routing path.
This problem is then a multi-commodity flow problem. We solve it with a column generation method similar to (but simpler than) CG-RR. Its solution is called \emph{NoCycleInfo} in the rest of this section.}

\begin{figure*}[!h]
	\centering
	\begin{subfigure}[b]{0.9\textwidth}
		\centering
		\includegraphics[width=\textwidth]{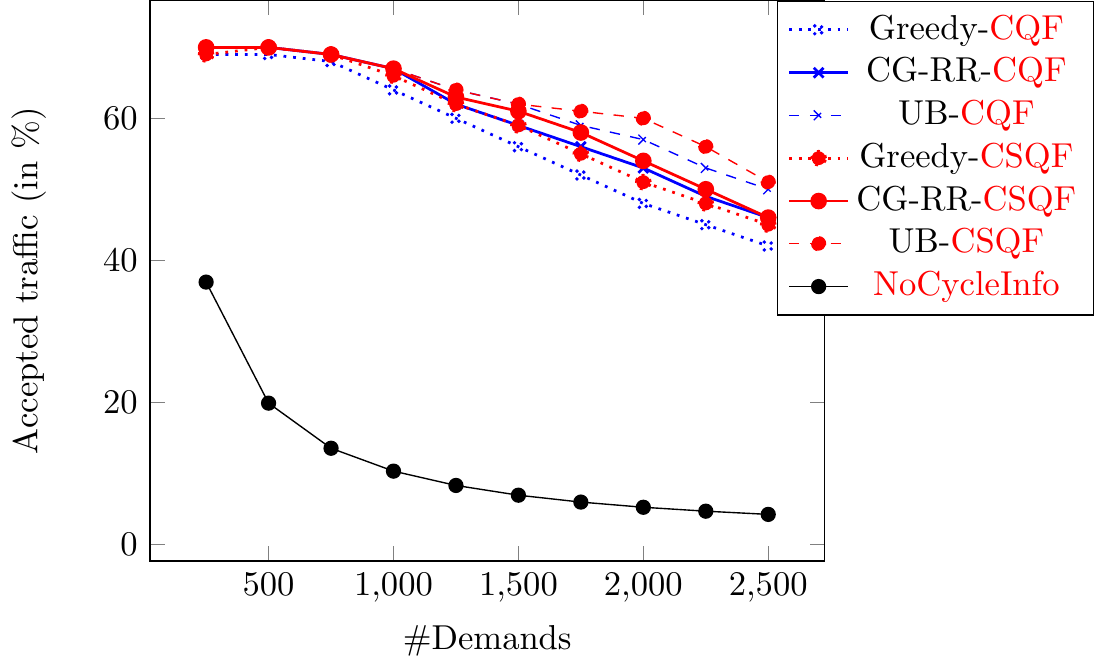}
		\caption{Accepted traffic for \emph{Greedy}, \emph{CG-RR} (with 2 and 3 queues for CQF and CSF, respectively), best Upper Bound (UB).}
		\label{fig:accTraff}
	\end{subfigure}
	\hfill
	\begin{subfigure}[b]{0.75\textwidth}
		\centering
		\includegraphics[width=\textwidth]{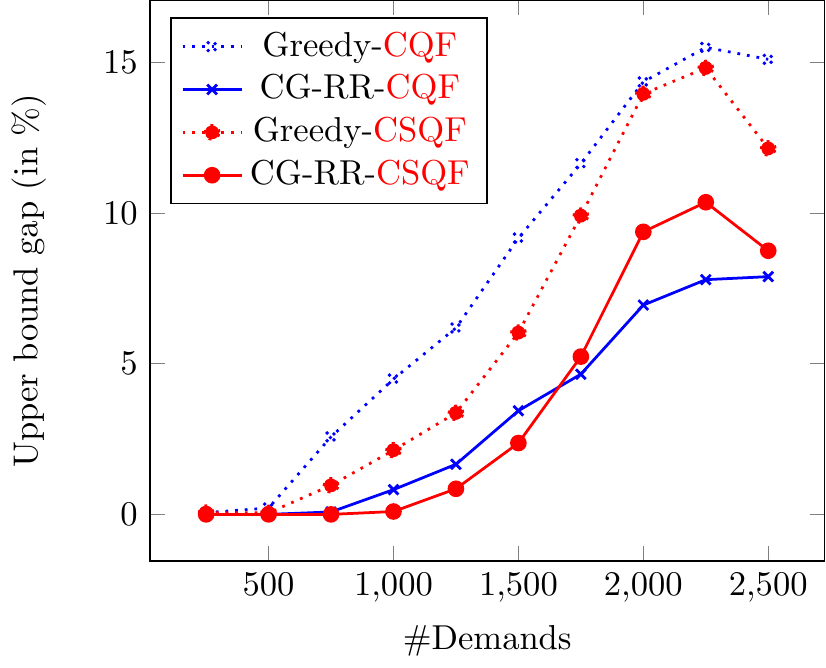}
		\caption{Gap to the best upper bound for \emph{Greedy} and \emph{CG-RR}.}
		\label{fig:gap}
	\end{subfigure}
	\label{fig:results1}
		\caption{Benchmarking results for \emph{CG-RR} and \emph{Greedy} on a realistic IPRAN scenario varying the number of demands and the number of queues (denoted \rev{CQF} and \rev{CSQF}, for 2 and 3 queues, respectively).}
\end{figure*}
\subsection{Results}

 
\rev{Figure~\ref{fig:accTraff} shows the accepted traffic (in \%) of the total demands of each instance for the different algorithms, together with the upper bounds of the CQF and the CSQF cases, derived from the respective linear relaxations. We observe first of all as expected that the use of algorithms adapted to CQF/CSQF are superior to the simple reservation of the  worst-case, represented by NoCycleInfo. Second of all, while the percentage of accepted traffic decreases with increasing demands, CSQF shows a superior acceptance rate than CQF. This is due to the flexibility introduced by the scheduling at each intermediate node.}

\rev{Comparing CQF and CSQF in more detail, Figure~\ref{fig:gap} shows the gap to the respective best upper bound coming from the linear relaxation in \emph{CG-RR}, for both \emph{Greedy} and \emph{CG-RR}}. We can see that for small amounts of demands (i.e., less than 500), both \emph{CG-RR} and \emph{Greedy} nearly give the optimal solution. When the traffic increases (around 1000 demands), \emph{CG-RR} still manages to get a solution equal to the upper bound (i.e., an optimal solution). For larger traffic, instead, both solutions plot an increasing gap because the linear relaxation provides an infeasible solution that accepts more traffic by splitting demands over multiple paths and multiple cycles. However, the real optimum lies in between the best upper bound and the integer solution found by \emph{CG-RR}. 
\emph{CG-RR} allows to provide a gap smaller than 10\% for all the considered traffic scenarios, both for \rev{CQF and CSQF}. The gap is slightly better for \rev{CQF} as the results provided by the linear relaxation are closer to the integer solution.
The use of 3 queues in \rev{CSQF (instead of 2 for CQF)} allows to accept more demands as it enables to postpone traffic with non-critical delay constraints. 

\begin{figure*}[!h]
	\centering
	\begin{subfigure}[b]{0.75\textwidth}
		\centering
		\includegraphics[width=\textwidth]{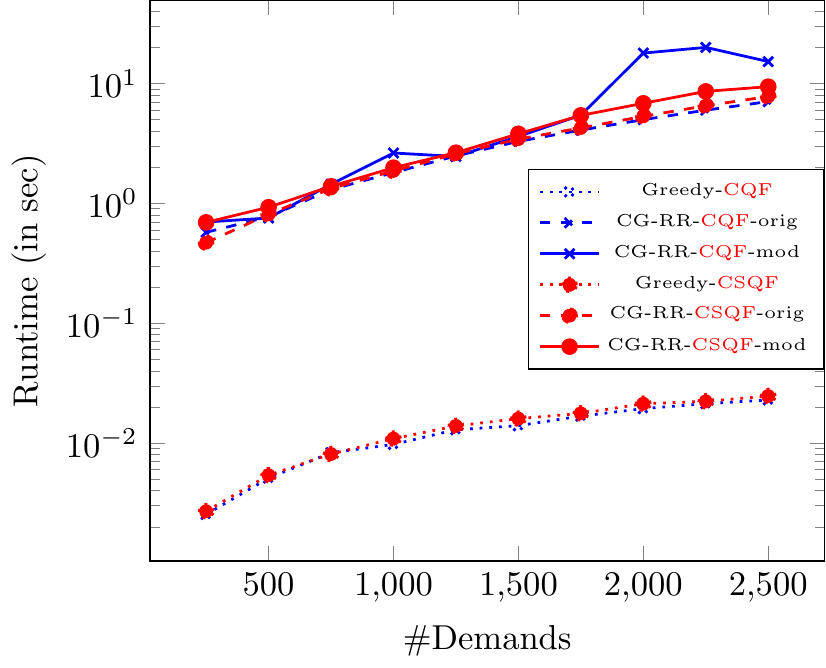}
		\caption{Execution time: \emph{Greedy}, \emph{CG-RR} w.\ original (-orig) and modified (-mod) constraints.}
		\label{fig:runTime}
	\end{subfigure}
	\hfill
	\begin{subfigure}[b]{0.75\textwidth}
		\centering
		\includegraphics[width=\textwidth]{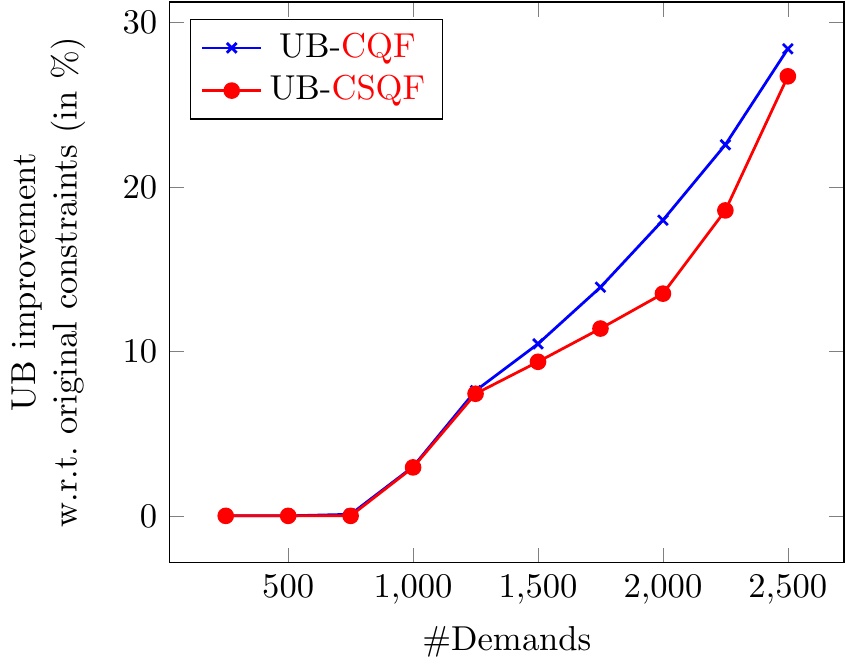}
		\caption{Improvement of the Upper Bound (UB) with modified constraints over original ones.}
		\label{fig:gapImprovement}
	\end{subfigure}
	\caption{Runtime and UB with modified constraints}
	\label{fig:results2}
\end{figure*}

As shown in Figure~\ref{fig:runTime}, \emph{CG-RR} provides a solution within a few seconds, which is quite reasonable for offline network planning.
While improving the solution by up to 5\% (see next paragraph), the reinforcement of the model leads to a marginal increase of the execution time by up to 30\% in the case of 3 queues. On the other hand, \emph{Greedy}, which is paying for a larger gap to the best upper bound, can provide a solution to the planning problem within hundreds of microseconds. \rev{To give idea about the memory used during computation, CG-RR needed on average about 500MB RAM to solve instances with 250 demands and 750 MB for 2500 demands.} \rev{As we focused on the rapid resolution of the planning problem, we did not consider arrivals and departures of demands in this evaluation scenario. However, from our results, we believe that the \emph{Greedy} algorithm is very suitable for online and ultra-fast demand acceptance (10 $\mu s$ per demand).}

The fact that we have an upper bound close to the solution provided by \emph{CG-RR} and \emph{Greedy} mainly depends on the reinforcement of constraints presented in Sec.~\ref{inequalities}. Figure~\ref{fig:gapImprovement} shows that for low traffic there is no significant improvement as the optimal solution is already provided. However, for larger traffic scenarios the improvement can be up to 30\% as the linear relaxation is closer to the integer solution. 
The reinforcement of constraints allows to produce a solution for the linear relaxation that is closer to the integer optimum. As the use of 2 queues reduces the possibility of splitting traffic, a better upper bound is found. 

\begin{figure*}[!h]
	\centering
	\begin{subfigure}[b]{0.75\textwidth}
		\centering
		\includegraphics[width=\textwidth]{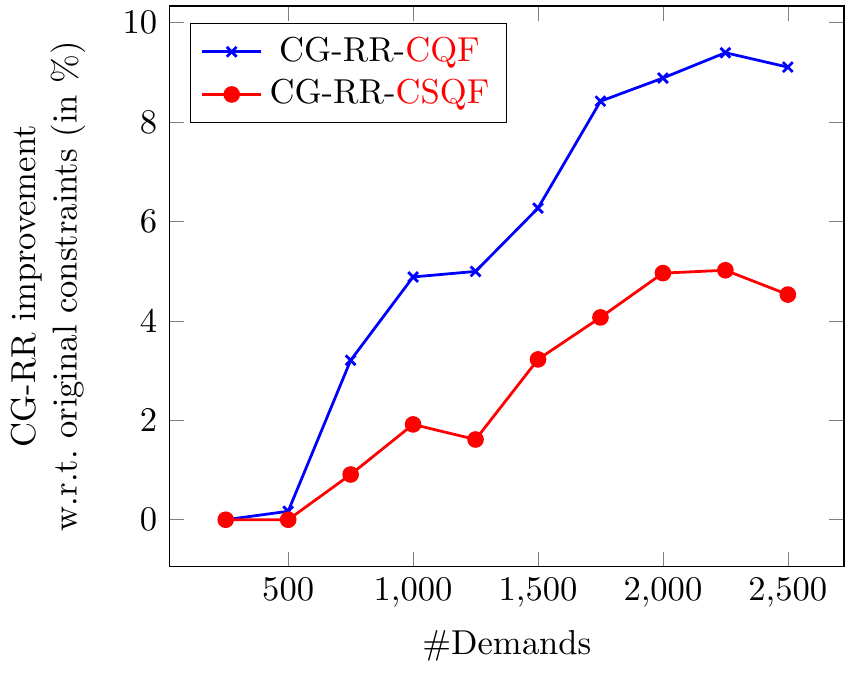}
		\caption{Improvement of \emph{CG-RR} solutions with modified constraints over original ones.}
		\label{fig:improvement}
	\end{subfigure}
	\hfill
	\begin{subfigure}[b]{0.75\textwidth}
		\centering
		\includegraphics[width=\textwidth]{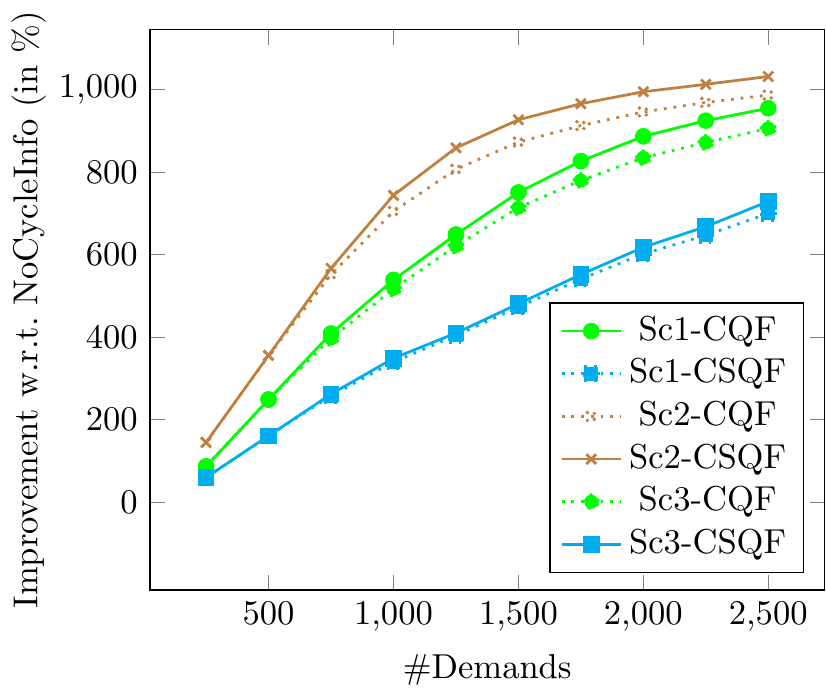}
		\caption{Performance of algorithms on different demand scenarios}
		\label{fig:scenarios}
	\end{subfigure}
\caption{CG-RR with modified constraints and different demand scenarios.}
	\label{fig:results3}
\end{figure*}

Figure~\ref{fig:improvement}  presents the improvement of \emph{CG-RR} with modified constraints over the case with original constraints in terms of accepted traffic. We can see that the improvement is up to 10\% \rev{for CQF} and up to 5\% for \rev{CSQF}. The improvement comes from the fact that the reinforcement model gets a linear relaxation closer to the integer solution.
As before, the smaller improvement for CSQF is due to the split of traffic in the linear relaxation over multiple paths and multiple cycles. 

\rev{Finally, we show in Fig.~\ref{fig:scenarios} the results of our CG-RR algorithm on the different load scenarios described in Sec.~\ref{sec:eval-setup}. For each scenario, we compare the traffic acceptance rate improvement (in \%) over the basic NoCycleInfo algorithm both for CQF and CSQF when the number of demands increases. We observe that the superiority of CQF/CSQF-adapted algorithms holds over a the variety of scenarios, particularly for larger instances. We also observe that the additional scheduling opportunities provided by CSQF always improve traffic acceptance compared to CQF.}

\section{Conclusion}
\label{conclusion}

In this paper we presented two algorithms for the joint routing and scheduling problem of time-triggered flows in large scale deterministic networks using CSQF. We formulated the problem as an extension of a multi-commodity flow problem and analyzed its NP-hardness. We proposed an effective solution based on column generation and dynamic programming. Thanks to the reinforcement of the model with valid inequalities, we improved the upper bound and the solution. On realistic IPRAN instances, we demonstrated that we reach an optimality gap smaller than 10\% in a few seconds. Finally, we also derived an ultra-fast adaptive greedy algorithm (10 $\mu s$ per demand) that can be used online flow admission at the cost of an extra 5\% gap when compared to our advanced solution based on column generation.

\rev{Future work along these lines may include the development of an approximation algorithm with guarantees on the integrality gap for the rounding phase. It may also include the development of an online algorithm based on recent primal dual methods~\cite{buchbinder2009online} to guarantee a certain competitive ratio.}

\bibliographystyle{IEEEtran}
\balance
\bibliography{biblio}

\end{document}